\documentclass[10pt,journal,compsoc]{IEEEtran}
\usepackage[utf8]{inputenc}
\usepackage{amsmath}
\usepackage{harpoon}
\usepackage{amssymb}
\usepackage{amsfonts}
\usepackage{subfig}
\usepackage{amsthm}
\usepackage{color}
\usepackage{autonum}
\usepackage{multicol}
\usepackage[textwidth=4.2cm,textsize=scriptsize]{todonotes}  
\LetLtxMacro{\todom}{\todo}
\renewcommand{\todo}[1]{\todom[inline]{#1}}

\newcommand{\set}[1]{\left\lbrace #1 \right\rbrace}

\renewcommand\vec{\mathbf}

\newcommand{\ignore}[1]{}

\usepackage[normalem]{ulem}

\newcommand{\va}{\vec{a}}
\newcommand{\vc}{\vec{c}}
\newcommand{\vd}{\vec{d}}

\newcommand{\vr}{\vec{r}}

\newcommand{\vo}{\vec{o}}
\newcommand{\vM}{\vec{M}}

\newcommand{\vx}{\vec{x}}

\newcommand{\vu}{\vec{u}}

\newcommand{\vb}{\vec{b}}
\newcommand{\bfr}{\mathbf{r}}
\newcommand{\bfp}{\mathbf{p}}

\newcommand{\Rp}{\mathbb{R}_{\geq 0}}

\def\longrightharpoonup{\relbar\joinrel\rightharpoonup}
\def\longleftharpoondown{\leftharpoondown\joinrel\relbar}
\def\longrightleftharpoons{\mathop{\vcenter{\hbox{\ooalign{\raise1pt\hbox{$\longrightharpoonup\joinrel$}\crcr\lower1pt\hbox{$\longleftharpoondown\joinrel$}}}}}}
\def\rxn{\mathop{\rightarrow}\limits}  

\newcommand{\slto}{\to^1}      

\usepackage{amsfonts}

\renewcommand{\emptyset}{\varnothing}

\newcommand{\pair}[2]{\left\langle #1 , #2 \right\rangle}

\newcommand{\N}{\mathbb{N}}

\newcommand{\R}{\mathbb{R}}

\newcommand{\calC}{{\cal C}}

\newcommand{\dom}{{\rm dom} \;}

\newtheorem{definition}{Definition}
\newtheorem{theorem}{Theorem}
\newtheorem{lemma}{Lemma}

\newtheorem{corollary}{Corollary}

\begin{document}

\title{Composable Rate-Independent Computation in Continuous Chemical Reaction Networks}

\author{Cameron Chalk, Niels Kornerup, Wyatt Reeves, David Soloveichik
\IEEEcompsocitemizethanks{\IEEEcompsocthanksitem The authors were with the University of Texas at Austin.\protect
\IEEEcompsocthanksitem C. Chalk and D. Soloveichik were supported in part by National Science Foundation grants CCF-1618895 and CCF-1652824.\protect\\}}

\IEEEtitleabstractindextext{%
\begin{abstract}
Biological regulatory networks depend upon chemical interactions to process information.  
Engineering such molecular computing systems is a major challenge for synthetic biology and related fields.
The chemical reaction network (CRN) model idealizes chemical interactions, allowing rigorous reasoning about the computational power of chemical kinetics.
Here we focus on function computation with CRNs, where we think of the initial concentrations of some species as the input and the equilibrium concentration of another species as the output. Specifically, we are concerned with CRNs that are rate-independent (the computation must be correct independent of the reaction rate law) and composable ($f\circ g$ can be computed by concatenating the CRNs computing $f$ and $g$).
Rate independence and composability are important engineering desiderata,
permitting implementations that violate mass-action kinetics, or even ``well-mixedness", 
and allowing the systematic construction of complex computation via modular design. 
We show that to construct composable rate-independent CRNs, it is necessary and sufficient to ensure that the output species of a module is not a reactant in any reaction within the module.
We then exactly characterize the functions computable by such CRNs as superadditive, positive-continuous, and piecewise rational linear.
Thus composability severely limits rate-independent computation unless more sophisticated input/output encodings are used.

\end{abstract}

}
\maketitle

\IEEEraisesectionheading{\section{Introduction}}
A ubiquitous form of biological information processing occurs in complex chemical regulatory networks in cells.
The formalism of chemical reaction networks (CRNs) has been widely used for modelling the interactions underlying such natural chemical computation. 
More recently CRNs have also become a useful model for designing synthetic molecular computation.
In particular, DNA strand displacement cascades can in principle realize arbitrary CRNs, thus motivating the study of CRNs as a programming language~\cite{cardelli2011strand,chen2013programmable,soloveichik2010dna,srinivas2017enzyme}.
The applications of synthetic chemical computation include reprogramming biological regulatory networks, as well as embedding control modules in environments that are inherently incompatible with traditional
electronic controllers for biochemical, nanotechnological, or medical applications.

The study of information processing within biological CRNs, as well the engineering of CRN functionality in artificial systems, motivates the exploration of the computational power of CRNs.
In general, CRNs are capable of Turing universal computation~\cite{fages2017strong}; however, we are often interested in restricted classes of CRNs which may have certain desired properties.
Previous work distinguished two programmable features of CRNs: the stoichiometry of the reactions and the rate laws governing the reaction speeds~\cite{rate-indep}.
As an example of computation by stoichiometry alone, 
consider the reaction $2X \rightarrow Y$.
We can think of the concentrations of species $X$ and $Y$ to be the input and output, respectively.
Then this reaction effectively computes $f(x) = \frac{x}{2}$, as in the limit of time going to infinity, the system converges to producing one unit of $Y$ for every two units of $X$ initially present.
The reason we are interested in computation via stoichiometry is that it is fundamentally \emph{rate-independent},
requiring no assumptions on the rate law (e.g., that the reaction occurs at a rate proportional to the product of the concentrations of the reactants).
This allows the computation to be correct independent of experimental conditions such as temperature, chemical background, or whether or not the solution is well-mixed.

Computation does not happen in isolation.
In an embedded chemical controller, inputs would be produced by other chemical systems, and outputs would affect downstream chemical processes.
Composition is easy in some systems (e.g.\ digital electronic circuits can be composed by wiring the outputs of one to the inputs of the other).
However, in other contexts composition presents a host of problems. 
For example, the effect termed retroactivity, which results in insufficient isolation of modules, has been the subject of much research in synthetic biology~\cite{del2008modular}. 
In this paper, we attempt to capture a natural notion of composable rate-independent computation, and study whether composability restricts computational power.

\begin{minipage}{.375\textwidth}
\begin{align}
&\textrm{\textbf{(a)}}\\
X_1 + X_2 &\rightarrow Y
\end{align}
\end{minipage}
\begin{minipage}{.48\textwidth}
\begin{align}
&\textrm{\textbf{(b)}}\\
X_1 &\rightarrow Z_1 + Y \\
X_2 &\rightarrow Z_2 + Y \\
Z_1 + Z_2 &\rightarrow K \\
Y + K &\rightarrow \emptyset.
\end{align}
\end{minipage}
\vspace{.5cm}

Above, we see two examples of rate-independent computation.
Example \textbf{(a)} shows $y = \min(x_1,x_2)$. The amount of $Y$ eventually produced will be the minimum of the initial amounts of $X_1$ and $X_2$, since the reaction will stop as soon as the first reactant runs out. 
Example \textbf{(b)} shows $y = \max(x_1,x_2)$. 
(The last reaction generates waste species that are not used in the computation; thus, we describe the products as an empty set.)
The amount of $Y$ eventually produced in reactions $1$ and $2$ is the sum of the initial amounts of $X_1$ and $X_2$. The amount of $K$ eventually produced in reaction $3$ is the minimum of the initial amounts of $X_1$ and $X_2$. Reaction $4$ subtracts the minimum from the sum, yielding the maximum.

Now consider how rate-independent computation can be naturally composed.
Suppose we want to compute $\min(\min(x_1,x_2),x_3)$.
It is easy to see that simple concatenation of two min modules (with proper renaming of the species) correctly computes this function:
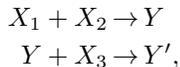
\begin{align}
    X_1 + X_2 &\rxn Y \\
    Y + X_3 &\rxn Y', 
\end{align}
where $Y'$ represents the output of the composed computation.
In contrast, suppose we want to compute $\min(\max(x_1,x_2),x_3)$.
Concatenating the modules yields:
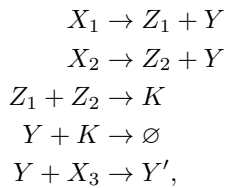
\begin{align}
    X_1 &\rightarrow Z_1 + Y \\
    X_2 &\rightarrow Z_2 + Y \\
    Z_1 + Z_2 &\rightarrow K \\
    Y + K &\rightarrow \emptyset \\
    Y + X_3 &\rightarrow Y',
\end{align}
where $Y'$ represents the output of the composed computation.
Observe that depending on the relative rates of reactions $4$ and $5$, the eventual value of $Y'$ will vary between $\min(\max(x_1,x_2),x_3)$ and $\min(x_1+x_2,x_3)$,
and the composition does not compute in a rate-independent manner.

Why is min composable, but max not? 
The problem arose because the output of the max module ($Y$) is consumed in both the max module and in the downstream min module. 
This creates a competition between the consumption of the output within its own module and the downstream module.

Towards modularity, we assume the two CRNs to be composed do not share any species apart from the interface between them (i.e., a species $Y$ representing the output of the first network is used as the species representing the input to the second network, and otherwise the two sets of species are disjoint).
We prove that to construct composable rate-independent modules in this manner, it is necessary and sufficient to ensure that the output species of a module is not a reactant in any reaction of that module.
We then exactly characterize the computational power of composable rate-independent computation.

Previously it was shown that without the composability restriction, rate-independent CRNs can compute arbitrary positive-continuous, piecewise rational linear functions~\cite{rate-indep}.
Positive-continuity means that the only discontinuities occur when some input goes from 0 to positive, 
and piecewise rational linear means that the function can by defined by a finite number of linear pieces (with rational coefficients).
Note that non-linear continuous functions can be approximated to arbitrary accuracy.\footnote{To approximate arbitrary continuous non-linear functions, piecewise linear functions are not sufficient, but rather we need piecewise affine functions (linear functions with offset). However, affine functions can be computed if we use an additional input fixed at $1$.}
We show that requiring the CRN to be composable 
restricts the class of computable functions to be superadditive functions; i.e., functions that satisfy: for all input vectors $\va, \vb,\, f(\va) + f(\vb) \leq f(\va + \vb)$.
This strongly restricts computational power: for example, subtraction or max cannot be computed or approximated in any reasonable sense.
In the positive direction, we show that any
superadditive, positive-continuous, piecewise rational linear function 
can be computed by composable CRNs in a rate-independent manner. 
Our proof is constructive, and we further show that unimolecular and bimolecular reactions are sufficient.

We note that different input and output encodings can change the computational power of rate-independent, composable CRNs.
For example, in the so-called \emph{dual-rail} convention, input and output values are represented by differences in concentrations of two species (e.g., the output is equal to the concentration of species $Y^+$ minus the concentration of $Y^-$).
Dual-rail simplifies composition---instead of consuming the output species to decrease the output value, a dual-rail CRN can produce $Y^-$---at the cost of greater system complexity.
Dual-rail CRNs can compute the full class of continuous, piecewise rational linear functions while satisfying rate-independence and composability~\cite{rate-indep}.
Note, however, that the dual-rail convention moves the non-superadditive subtraction operation to ``outside'' the system, and converting from a dual-rail output to a direct output must break composability.

\section{Preliminaries}
Let $\N$ and $\R$ denote the set of nonnegative integers and the set of real numbers, respectively. The set of the first $n$ positive integers is denoted by $[n]$.
Let $\R_{\geq 0}$ be the set of nonnegative real numbers, and similarly $\R_{> 0}$ be the set of positive real numbers.
If $\Lambda$ is a finite set (in this paper, of chemical species), we write $\R^\Lambda$ to denote the set of functions $f:\Lambda \to \R$, and similarly for $\Rp^\Lambda$, $\N^\Lambda$, etc.
Equivalently, we view an element $\vc\in A^\Lambda$ as a vector of $|\Lambda|$ elements of $A$, each coordinate ``labeled'' by an element of $\Lambda$.
Given a function $f : A \rightarrow B$, we use $f|_C$ to denote the restriction of $f$ to the domain $C$. We also use the notation $\vc \upharpoonright \Delta$ to represent $\vc$ projected onto $\Rp^\Delta$. Thus, $\vc  \upharpoonright \Delta = \vec{0}$ iff $(\forall S\in\Delta)\ \vc(S)=0$.
If $\Delta \subseteq \Lambda$, we view a vector $\vc \in \Rp^\Delta$ equivalently as a vector $\vc \in \Rp^\Lambda$ by assuming $\vc(S)=0$ for all $S \in \Lambda \setminus \Delta.$

\subsection{Chemical reaction networks}
We will start by defining the notation used to describe chemical reactions.


\begin{definition}
Given a finite set of chemical species $\Lambda$, a \emph{reaction} over $\Lambda$ is a pair $\alpha = \langle \bfr,\bfp \rangle \in \N^\Lambda \times \N^\Lambda$, specifying the stoichiometry of the reactants and products, respectively.\footnote{As we are studying CRNs whose output is independent of the reaction rates, we leave the rate constants out of the definition.} 
\end{definition}

\noindent In this paper, we assume that $\bfr \neq \vec{0}$, i.e., we have no reactions of the form $\emptyset \to \ldots$.
For instance, given $\Lambda=\{A,B,C\}$, the reaction $A+2B \to A+3C$ is the pair $\pair{(1,2,0)}{(1,0,3)}$.

\begin{definition}
A \emph{(finite) chemical reaction network (CRN)} is a pair $\calC=(\Lambda,R)$, where $\Lambda$ is a finite set of chemical \emph{species},
and $R$ is a finite set of reactions over $\Lambda$. 
\end{definition}

\noindent
Next we map language about chemical reaction networks to formal definitions and notation.

\begin{definition}
A \emph{state} of a CRN $\calC=(\Lambda,R)$ is a vector $\vc \in \Rp^\Lambda$.
\end{definition}

\begin{definition}
For any $\vc \in \Rp^\Lambda$ and any $S\in \Lambda$, $\vc(S)$ is the \emph{concentration} of $S$ in $\vc$.
\end{definition}
\begin{definition}
For any $\vc \in \Rp^\Lambda$, the set of species \emph{present} in $\vc$ (denoted by $[\vc]$) is $\{S \in \Lambda \ |\ \vc(S) > 0 \}$.
\end{definition}

\begin{definition}
Given a state $\vc$ and reaction $\alpha=\pair{\bfr}{\bfp}$, we say that $\alpha$ is \emph{applicable} in $\vc$ if $[\bfr] \subseteq [\vc]$ (i.e., $\vc$ contains positive concentration of all of the reactants).
\end{definition}

\begin{definition}
A reaction \emph{produces} (\emph{consumes}) a species $S$ if $S$ appears as a product (reactant).\footnote{Note that typically a catalyst species is not considered consumed nor produced by a catalytic reaction. 
For simplicity, as our definition suggests, we say it is both produced and consumed.}
\end{definition}

\subsection{Reachability and stable computation}
We now follow \cite{rate-indep} in defining rate-independent computation in terms of reachability between states (this treatment is in turn based on the notion of ``stable computation'' in distributed computing~\cite{AngluinAER2007}).
Intuitively, we say a state is ``reachable'' if some rate law can take the system to this state. 
For computation to be rate-independent, 
since unknown rate laws might take the system to any reachable state, the system must be able to reach the correct output from any such reachable state.



To define the notion of reachability,
a key insight of~\cite{rate-indep} allows one to think of reachability via a sequence of straight line segments.
This may be unintuitive, since mass-action\footnote{Although the formal definition of mass-action kinetics is outside the scope of this paper, we remind the reader that a CRN with rate constants on each reaction define a system of ODEs under mass-action kinetics.
For example, the two reactions
$A + B \rightarrow A + C$ and
$C + C \rightarrow B$ correspond to the following ODEs:
    $\dot{a} = 0$,
    $\dot{b} = k_2c^2 - k_1ab$, and
    $\dot{c} = k_1ab - 2k_2c^2$,
where $a, b$, and $c$ are the concentrations of species $A, B$, and $C$ over time and $k_1$, $k_2$ are the rate constants of the reactions.}  and other rate laws trace out smooth curves.
However, a number of properties are shown which support straight-line reachability as an interpretation which includes mass-action 
reachability as well as reachability  under other rate laws. 


\begin{definition}
Let $\calC$ be a CRN defined by $(\Lambda, R)$. The linear transformation $\vM: \R^R \to \R^\Lambda$ that maps from the unit vector representing a reaction to the net change in species caused by that reaction is the \emph{stoichiometry matrix} for $\calC$. 
\end{definition}
\noindent Note that we can intuitively think of $\vM$ being a matrix where the columns represent the net change in species caused by each reaction. Under this representation, observe that entries in $\vM$ will be negative when more of a reactant is consumed than is produced in a reaction. Observe that the image of $\vM$ represents the possible changes in a state that can occur via the reactions in $R$. We will formalize this notion with the next few definitions.
\begin{definition}
For a CRN with the reactions $R$, we say that any vector $\vu \in \Rp^R$ is a \emph{flux vector}. We use $[\vu]$ to denote the set $\{r\ |\  \vu(r) > 0\}$. We say that $\vu$ is \emph{applicable} at a state $\vc$ if every reaction in $[\vu]$ is applicable at $\vc$.
\end{definition}
\begin{definition}
For a CRN with species $\Lambda$ and stoichiometry matrix $\vM$, we say a state $\vd \in \R_{\geq 0}^\Lambda$ is \emph{straight-line reachable} from $\vc$, written $\vc \slto \vd$, or more precisely as $\vc \to_{\vu} \vd$, if there is a applicable flux vector $\vu$ such that $\vc + \vM \vu = \vd$.
\end{definition}
\noindent Intuitively, a single segment means running the reactions applicable at $\vc$ at a constant (possibly 0) rate specified by $\vu$ to get from $\vc$ to $\vd$. Since applying a flux vector can change the set of species present, $\vec{a} \slto \vec{b}$ does not imply that $\vec{a}$ and $\vec{b}$ have the same set of applicable reactions. Therefore there can be a state $\vec{c}$ that is straight-line reachable from $\vec{b}$ but not from $\vec{a}$. This leads us to our next definition.
\begin{definition}
We say state $\vd$ is \emph{1-segment reachable} from $\vc$ if it is straight line reachable. We say a state $\vd$ is \emph{$l$-segment reachable} if there is a state $\vd'$ that is $(l-1)$-segment reachable from $\vc$ such that $\vd' \slto \vd$.
\end{definition}
\noindent Generalizing to an arbitrary number of segments, we obtain our general notion of reachability below.
Note that by the definition of straight-line reachability, only applicable reactions occur in each segment.
The definition of reachability is closely related to exploring the ``stoichiometric compatibility class'' of the initial state~\cite{feinberg1974dynamics}.

\begin{definition}\label{defn-reachable-segment}
A state $\vd$ is \emph{reachable} from $\vc$, written $\vc \rightarrow \vd$, if $\exists l\in\N$ such that $\vd$ is $l$-segment reachable from $\vc$.
We denote the set of states reachable from $\vec{c}$, i.e., $\{\vec{d} \ | \ \vec{c}\rightarrow \vec{d} \}$, as $\mathrm{Post}(\vec{c})$.
\end{definition}



We think of state $\vd$ as being reachable from state $\vc$ if there is a ``reasonable'' rate law that takes the system from $\vc$ to $\vd$.
As desired, previous work showed that if state $\vd$ is reached from $\vc$ via a mass-action trajectory, it is also segment-reachable.

\begin{lemma}[Proven in~\cite{rate-indep}]  \label{lem:mass-action-reachable}
   If $\vd$ is mass-action reachable from $\vc$, then $\vc \rightarrow \vd$.
\end{lemma}

We can now use reachability to  formally define rate-independent computation.
%
%
%
\begin{definition}\label{def:CRC}
A \emph{chemical reaction computer (CRC)} is a tuple $\calC = (\Lambda,R,\Sigma,Y)$, where $(\Lambda,R)$ is a CRN, $\Sigma \subset \Lambda$, written as $\Sigma = \{X_1,\ldots,X_n\}$, is the \emph{set of input species}, and $Y \in \Lambda \setminus \Sigma$ is the \emph{output species}.
\end{definition}
\noindent For simplicity, assume a canonical ordering of $\Sigma=\{X_1,\ldots,X_n\}$ so that a vector $\vx\in\Rp^n$ (i.e., an input to $f$) can be viewed equivalently as a state $\vx\in\Rp^\Sigma$ of $\calC$ (i.e., an input to $\calC$).



\begin{definition}
A state $\vo \in \Rp^\Lambda$ is \emph{output stable} if, for all $\vo'$ such that $\vo \to \vo'$, $\vo(Y) = \vo'(Y)$,  i.e., once $\vo$ is reached, no reactions can change the concentration of the output species $Y$.
\end{definition}

\begin{definition} \label{def:direct-computation}
Let $f:\Rp^n \to \Rp$ be a function and let $\calC$ be a CRC.
We say that $\calC$ \emph{stably computes} $f$ if, for all $\vx \in \Rp^n$ and all $\vc$ such that $\vx \to \vc$, there exists an output stable state $\vo$ such that $\vc \to \vo$ and $\vo(Y) = f(\vx)$.
\end{definition}

We can intuitively justify the above definition of reachability and stable computation as capturing the class of computation that is independent of the rate law.
The output stable states are exactly those in which the output cannot be changed by a rate law chosen by an adversary.
If a chemical reaction network does not stably compute a function, 
then some rate law can take the system to a state from which an output stable state is not reachable (including by mass-action by Lemma~\ref{lem:mass-action-reachable}).

The results herein extend easily to functions $f:\R^n \to \R^l$, i.e., whose output is a vector of $l$ real numbers.
This is because such a function is equivalently $l$ separate functions $f_i:\R^n\to\R$.

Also note that initial states contain only the input species $\Sigma$; other species must have initial concentration 0.
Section~\ref{sec:initial_context} discusses how allowing some initial concentration of non-input species affects computation.

\subsection{Composability}
In this section we define the composition of CRCs and formally relate composability to a CRC not using its output species as a reactant (output-obliviousness).
We show that output-oblivious CRCs are composable, and that any composable CRC can be reduced (simply by removing reactions) to an output-oblivious form.

We define the composition of two CRCs intuitively
as the concatenation of their chemical reactions, such that the output species of the first is the input species of the second:
\begin{definition}
Given two CRCs $\calC_1 = (\Lambda_1,R_1,\Sigma_1,Y_1)$ and $ \calC_2 = (\Lambda_2,R_2,\Sigma_2,Y_2)$, consider $\calC_2' = (\Lambda_2',R_2',\Sigma_2',Y_2')$ constructed by renaming species of $\calC_2$ such that $\Lambda_1\cap\Lambda_2' = \{Y_1\}$ and $Y_1 \in \Sigma_2'$.
The \emph{composition} of $\calC_1$ and $\calC_2$ is the CRC  $\calC_{2\circ1} = (\Lambda_1\cup\Lambda_2', R_1\cup R_2', \Sigma_1 \cup \Sigma_2' \setminus \{Y_1\}, Y_2')$.
In other words, the composition is constructed by concatenating $\calC_1$ and $\calC_2$ such that their only interface is the output species of $\calC_1$, used as the input for $\calC_2$.
\end{definition}

\noindent We say two CRCs are composable if they stably compute the composition of their functions when composed:

\begin{definition}
A CRC $\calC_1$ which stably computes $f_1$ is \emph{composable} if $\forall \calC_2$ stably computing $f_2$, $\calC_{2\circ1}$ stably computes $f_2\circ f_1$.
\end{definition}%
%
\noindent We want to relate composability to the property that a CRC does not use its output species as a reactant:
\begin{definition}
We call a CRC $(\Lambda,R,\Sigma,Y)$ \emph{output-oblivious} if $Y$ does not appear as a reactant in $R$.
\end{definition}

\noindent 
For simplicity, we focus on single-input, single-output CRCs, but these results easily generalize to multiple input and output settings.

For the proof that the output-oblivious condition is sufficient to guarantee composability, we formalize the idea that the composed CRCs act independently, and do not interfere with each other's execution.
In Lemmas~\ref{lem:spliting-reordering}~and~\ref{lem:output-oblivious-splitting} we show how this independence can be used to ``reorder'' the sequence of reactions of a CRC in way that preserves the concentrations in the final state.
In Lemma~\ref{lem:not-using-out-imp-comp}, we take an output-oblivious CRC, compose another CRC downstream, and reorder any sequence of reactions of the composed CRC into a sequence which we can easily argue must have stably computed as expected.



\begin{definition}
A flux vector $\vec{u}$ \emph{produces} (\emph{consumes}) a species $S$ if there is an $\vr \in [\vec{u}]$ such that $S$ is a product (reactant) of $\vr$.
\end{definition}
\begin{definition}
A flux vector $\vec{u}_1$ is \emph{independent of} a flux vector $\vec{u}_2$ if $\vec{u}_1$ does not consume any species that are produced or consumed by $\vec{u}_2$. 
\end{definition}
\begin{lemma}\label{lem:spliting-reordering}
In a CRC $\calC = (\Lambda, R)$ if flux vector $\vec{u}_1$ is independent of flux vector $\vec{u}_2$ then:
\begin{enumerate}
    \item If $\va \to_{\vu_2} \vb \to_{\vu_1} \vc$, then $\va \to_{\vu_1 + \vu_2} \vc$.
    \item If $\va \to_{\vu_1 + \vu_2} \vc$, then there is a state $\vb$ such that $\va \to_{\vu_1} \vb \to_{\vu_2} \vc$.
    \item If $\va \to_{\vu_1} \vb$ and $\va \to_{\vu_2} \vc$, then there is a state $\vd$ such that $\vc \to_{\vu_1} \vd$.
\end{enumerate}
\end{lemma}

\begin{proof}
For $1$, since $\vec{u}_1$ is independent of $\vec{u}_2$, we know that $\vec{u}_2$ cannot produce any of the species necessary to make $\vec{u}_1$ applicable. Since $\vec{u}_1$ was applicable at $\vec{b}$, we know that $\vec{u}_1$ must be applicable at $\vec{a}$, so $\vec{a} \to_{\vec{u}_1 + \vec{u}_2} \vec{c}$.

For $2$, we need to show that $\vb \in \Rp^\Lambda$ (has nonnegative concentrations)
and that $\vu_2$ is applicable at $\vb$.
Consider any species $S \in \Lambda$ such that $S$ is not produced in $\vu_2$. Since $\vb = \vc - \vM\vu_2$ and $\vc(S) \geq 0$, we know that $S$ has nonnegative concentration at $\vb$. Now consider any species $S \in \Lambda$ such that $S$ is produced in $\vu_2$. Since $\vu_2$ produces $S$, we know that $\vu_1$ must not consume it because $\vec{u}_1$ is independent of $\vec{u}_2$. Since $\vb = \va + \vM\vu_1$, we can conclude that $S$ must have nonnegative concentration at $\vb$. 
Therefore we can conclude that $\vb \in \R_{\geq 0}^\Lambda$.
To see that $\vu_2$ is applicable at $\vb$, first observe that 
$\vu_2$ is applicable at $\va$. Then, since $\vu_1$ is independent of $\vu_2$, it follows that  $\vu_2$ must also be applicable at $\vb$. 

For $3$,
we want to show that $\vu_1$ is still applicable at $\vc$ and that $\vd = \vc + M\vu_1$ is in $\Rp^\Lambda$ (has nonnegative concentrations).
$\vu_1$ is still applicable at $\vc$ since $\vu_1$ does not consume species involved in $\vu_2$.
If $\vd$ had negative concentration on species $S$, that species must have been consumed by $\vu_1$ since $\vc$ has nonnegative concentrations.
Since $\vu_2$ does not consume any species consumed by $\vu_1$ and $\vc = \va + \vM\vu_1$, then negative $\vd(S)$ implies negative $\vc(S)$, which contradicts that $\vc \in \Rp^\Lambda$, so $\vd \in \Rp^\Lambda$.

\end{proof}



\begin{lemma}\label{lem:output-oblivious-splitting}
Given two output-oblivious CRCs $\calC_1$ and $\calC_2$,
consider the composition CRC $\calC_{2\circ1}$.
If $\vec{c} \to \vec{d}$ then there is $\vec{b}$ such that $\vec{c} \to \vec{b} \to \vec{d}$, where $\vec{c} \to \vec{b}$ only uses reactions from $\calC_1$ and $\vec{b} \to \vec{d}$ only uses reactions from $\calC_2$.
\end{lemma}

\begin{proof}
Let $\{\vec{v}_1, \ldots, \vec{v}_n\}$ be the flux vectors such that $\vec{c} \to_{\vec{v}_1} \vec{c}_1 \to_{\vec{v}_2} \ldots \to_{\vec{v}_n} \vec{d}$. We can write $\vec{v}_i = \vec{u}_{1, i} + \vec{u}_{2, i}$, where $\vec{u}_{1,i}$ corresponds to the reactions in $\calC_1$ and $\vec{u}_{2, i}$ corresponds to the reactions in $\calC_2$. Since $\calC_1$ is output-oblivious, we know that every $\vu_{1,i}$ is independent of every $\vu_{2,j}$ and thus we can apply Lemma~\ref{lem:spliting-reordering} item $2$ to see that $\vec{c} \to_{\vec{u}_{1, 1}} \vec{b}_1 \to_{\vec{u}_{2,1}} \vec{c}_1 \to_{\vec{u}_{1, 2}} \vec{b}_2 \ldots \to_{\vec{u}_{2, n}} \vec{d}$. By repeatedly applying Lemma~\ref{lem:spliting-reordering} items $1$ and $2$, we can then rearrange the sequence of reactions so that each $\vec{u}_{1, i}$ precedes each $\vec{u}_{2, j}$ to get $\vec{c} \to_{\vec{u}_{1, 1}} \vec{b}_1 \to_{\vec{u}_{1,2}} \ldots \to_{\vec{u}_{1,n}} \vec{b} \to_{\vec{u}_{2, 1}}  \ldots \to_{\vec{u}_{2, n}} \vec{d}$. 
\end{proof}

\begin{lemma}\label{lem:not-using-out-imp-comp}
Output-oblivious
CRCs are composable.
\end{lemma}
\begin{proof}
Consider the composition $\calC_{2\circ1} =(\Lambda,R,\Sigma,Y)$ of two CRCs $\calC_1=(\Lambda_1,R_1,\Sigma_1,Y_1)$ and $\calC_2=(\Lambda_2,R_2,\Sigma_2,Y_2)$ that stably compute $f_1$ and $f_2$ respectively, and consider an input $x \in \R_{\geq0}^{\Sigma}$. Consider some state $\vec{c}$ reached from $\vec{x}$ in $\calC_{2 \circ 1}$. Let $\vec{b}$ be as in Lemma~\ref{lem:output-oblivious-splitting}, so $\vx$ $\to_{\vec{u}_{1,1}} \ldots \to_{\vec{u}_{1,n}} \vec{b} \to_{\vec{u}_{2,1}} \ldots \to_{\vec{u}_{2,n}}$ $\vc$, where $\vec{r}_i = (\vu_{i, 1}, \ldots, \vu_{i, n})$ is a sequence of flux vectors with $\bigcup_j [\vu_{i,j}] \subseteq R_i$. Since $\calC_1$ stably computes $f_1$, we know that there is some $\calC_1$-output stable state $\vec{o}_1$ reachable from $\vec{b}$ using a series of flux vectors $\vec{r} = (\vu_1, \vu_2 \ldots \vu_k)$ such that $\bigcup_{i}[\vu_i] \subseteq R_1$. Since $\vec{r}_2$ only uses reactions from $\calC_2$ and $\calC_1$ is output-oblivious, every flux vector in $\vec{r}$ is independent of every flux vector in $\vec{r}_2$, so by Lemma~\ref{lem:spliting-reordering} item $3$ we know the sequence of flux vectors $\vr$ is applicable starting at $\vec{c}$. Let $\vec{a}$ be such that $\vec{c} \to_{\vec{u}_{1}} \ldots \to_{\vec{u}_{k}} \vec{a}$. Then applying Lemma~\ref{lem:spliting-reordering} items $1$ and $2$ repeatedly to the flux vectors in $\vec{r}$ and $\vec{r}_2$, we see that $\vec{o}_1 \to_{\vec{u}_{2,1}} \ldots \to_{\vec{u}_{2,n}} \vec{a}$. Since $\calC_2$ stably computes $f_2$, since $\vec{o}_1(Y_1) = f(x)$, and since $\vec{a}$ is reachable from $\vec{o}_1$ only using reactions in $\calC_2$, there must be some $\vec{o}_2$ that is $\calC_2$-output stable such that $\vec{a} \to \vec{o}_2$ and $\vec{o}_2(Y_2) = f_2 \circ f_1(x)$. We know that $\vec{o}_2$ is reachable from $\vec{c}$ since $\vec{c} \to \vec{a} \to \vec{o}_2$. Finally, since $\vec{o}_1$ is $\calC_1$-output stable, reactions from $\calC_1$ cannot change the concentrations of species in $\vec{o}_2\upharpoonright{\Lambda_2}$, so if $\vec{o}_2 \to \vec{y}$, then restricting to $\calC_2$ we find $\vec{o}_2\upharpoonright {\Lambda_2} \to_{\calC_2} \vec{y} \upharpoonright {\Lambda_2}$. Since $\vec{o}_2$ is $\calC_2$-output stable we see that $\vec{y}(Y_2) = \vec{o}_2(Y_2)$, so $\vec{o}_2$ is $\calC_{2 \circ 1}$-output stable.
\end{proof}

Next we show that the output-oblivious condition is effectively necessary for composition.
Technically, there are CRCs which are not output-oblivious but are composable. However, we show that for such CRCs, we can remove reactions until they are output-oblivious, resulting in a CRC which is still composable and computes the same function.
Thus, characterizing what is computable by output-oblivious CRCs does characterize the class of functions computable by composable CRCs.

To prove this, we will want to reach a state that has used up its capacity to produce more of some species $S$. Intuitively this can be done by producing the maximum amount of $S$ possible by the CRC. 
However, in general this notion is ill-defined.\footnote{In the CRN $X \to Z, X + Z \to S + Z$, from an initial state $\vx$ consisting of only species $X$, note that $\max_{\vc\ \mid\ \vx \to \vc}\vc(S)$ is not well-defined.
To see this, note that some $\epsilon$ amount of $X$ must convert to $Z$ to catalyze the second reaction; we can always choose a smaller $\epsilon'$ to produce more $S$.
}
Lemma~\ref{lem:max} proves that there is a state with maximal amount of $S$. In other words, 
from any state we can always reach a state where afterwards it is impossible to increase the amount of $S$. The proof of Lemma~\ref{lem:max} is left to the appendix.

\begin{lemma}\label{lem:max}
For any state $\vec{c}$ and any species $S$, if the amount of $S$ present in any state reachable from $\vec{c}$ is bounded above, there is a state $\vec{d}$ reachable from $\vec{c}$ such that for any state $\vec{a}$ reachable from $\vec{d}$, we know that $\vec{a}(S) \le \vec{d}(S)$. 
\end{lemma}


\begin{lemma}\label{lem:comp-implies-not-using-out}
If a CRC $\calC$ stably computes $f$ and is composable, then we can remove all reactions where the output species appears as a reactant, and the resulting output-oblivious CRC will still stably compute $f$.
\end{lemma}

\begin{proof}
Let $\calC_1 = (\Lambda_1,R_1,\Sigma_1,Y_1)$ be a composable CRC stably computing some function $f$. Let $\calC_0 \subseteq \calC_1$ be the CRN obtained by removing all of the reactions that consume $Y_1$ from $\calC_1$. We would like to show that $\calC_0$ stably computes $f$. Suppose we compose $\calC_1$ with $\calC_2$ consisting of only the reaction $Y_1 \to Y_2$ with input species $Y_1$ and output species $Y_2$.
Since $\calC_1$ is composable and $\calC_2$ stably computes the identity function, the resulting CRN $\calC_{2\circ1}$ must stably compute $f$. For any input vector $\vec{x}$ consider a state $\vec{c}$ reachable from $\vec{x}$.

Assume that $\calC_0$ could reach a state $\vec{d}$ from $\vec{c}$ where $\vec{d}(Y_1) > f(\vec{x})$. Then $\calC_1$ would not be composable because $\calC_{2\circ1}$ can also reach $\vec{d}$ and then applying the reaction in $\calC_2$ to convert all $Y_1$ into $Y_2$ gives us a state $\vec{d'}$ with $\vec{d'}(Y_2) > f(\vec{x})$. Since there is no reaction in $\calC_{2\circ1}$ that consumes $Y_2$ there is no output stable state reachable from $\vec{d'}$ that computes $f(\vec{x})$. Therefore there is no state $\vec{d}$ such that $\vec{d}(Y_1) \geq f(\vec{x})$ and 
$\vec{d}$ is reachable from $\vec{c}$ via reactions of $\calC_0$. 

Since this implies that the amount of $Y_1$ in any state reachable from $\vec{c}$ is bounded, we can apply Lemma~\ref{lem:max} to say that there is a state $\vec{d}$ such that $\vec{c} \to \vec{d}$ and for any state $\vec{a}$ reachable from $\vec{d}$ we know $\vec{a}(Y_1) \le \vec{d}(Y_1) $. Since $\calC_0$ has no reactions that consume $Y_1$, this is an output stable state of $\calC_0$. Now, consider the state $\vec{b}$ in $\calC_{2\circ1}$ obtained by converting all $Y_1$ in $\vec{d}$ into $Y_2$. Observe that if there were a way to produce $Y_1$ from $\vec{b}$, then there would be a state in $\calC_0$ reachable from $\vec{d}$ that contained more $Y_1$. Since there are no reactions in $\calC_{2\circ 1}$ that consume $Y_2$ and no reactions that produce $Y_1$, we know that $\vec{b}$ is an output stable state. Since $\calC_{2 \circ 1}$ stably computes $f(\vec{x})$, we know $\vec{b}(Y_2) = f(\vec{x})$. Thus we can conclude that $\vec{d}(Y_1) = f(\vec{x})$ and $\calC_0$ stably computes $f(\vec{x})$.
\end{proof}

To allow composition of multiple downstream CRCs, we can use the reaction $Y \rightarrow Y_1 + \ldots + Y_n$ to generate $n$ ``copies'' of the output species $Y$, such that each downstream module uses a different copy as input.
Additionally, if the downstream module is output-oblivious, then the composition is also output-oblivious and thus the composition is composable.
These observations allow complex compositions of modules, and will be used in our constructions in Section~\ref{sec:thm_if}.

\section{Functions Computable by Composable CRNs}\label{sec:functions}
Here we give a complete characterization of the functions computable by composable CRNs as \emph{superadditive}, \emph{positive-continuous}, and \emph{piecewise rational linear}.

\begin{definition}
A function $f : \R^n \rightarrow \R^l$ is \emph{superadditive} iff $\forall \vec{a}, \vec{b} \in \R^n, f(\vec{a}) + f(\vec{b}) \leq f(\vec{a} + \vec{b})$.
\end{definition}

Note that superadditivity implies monotonicity in our case, since the functions computed must be nonnegative.
As an example, we show that the $\max$ function is not superadditive:

\begin{lemma}
The function $\max(x_1, x_2)$ is not superadditive.
\end{lemma}
\begin{proof}
Pick any $x_1, x_2 > 0$. Observe that $\max(x_1,0) + \max(0,x_2) = x_1 + x_2$. But since $x_1$ and $x_2$ are both positive, we know that $x_1 + x_2 > \max(x_1,x_2)$. Thus max is not superadditive and by Lemma~\ref{lem:superadditive} there is no composable CRN which stably computes max.
\end{proof}

\begin{definition}
A function $f : \R^n_{\geq0} \rightarrow \R^l$ is \emph{positive-continuous} if for all $U \subseteq [n]$, $f$ is continuous on the domain $D_U = \{\ x\in\R^n_{\geq0}\ |\ (\forall i \in [n])$,  $x(i) > 0 \iff i \in U \}$.
I.e., $f$ is continuous on any subset $D \subset \R^n_{\geq0}$ that does not have any coordinate $i\in[n]$ that takes both zero and positive values in $D$.
\end{definition}

Next we give our definition of piecewise rational linear.
One may (and typically does) consider a restriction on the domains selected for the pieces, however this restriction is unneccesary in this work, particularly because the additional constraint of positive-continuity gives enough restriction.

\begin{definition}\label{def:pw-rational-linear}
A function $f : \R^n \rightarrow \R$ is \emph{rational linear} if there exists $a_1,\ldots, a_n \in \mathbb{Q}$ such that $f(x) = \sum_{i=1}^n a_ix(i)$.
A function $f : \R^n \rightarrow \R$ is \emph{piecewise rational linear} if there is a finite set of partial rational linear functions $f_1,\ldots,f_p : \R^n \rightarrow \R$ with $\bigcup_{j=1}^p \dom f_j = \R^n$, such that for all $j \in [p]$ and all $x \in \dom f_j$, $f(x) = f_j(x)$. We call $f_1,\ldots,f_p$ the \emph{components} of $f$.
\end{definition}

The following is an example of a superadditive, positive-continuous, piecewise rational linear function:

\begin{IEEEeqnarray}{rCl}
f(\vec{x}) = \begin{cases}
      x_1 + x_2 & x_3 > 0 \\
      \min(x_1,x_2) & x_3 = 0.
   \end{cases}\label{eq:example}
\end{IEEEeqnarray}

The function is superadditive since for all input vectors $\vec{a} = (a_1, a_2, a_3)$, $\vec{b} = (b_1, b_2, b_3)$, there are three cases: \textbf{(1)} $a_3 = b_3 = 0$, in which case both input vectors compute $\min$ which is a superadditive function; \textbf{(2)} $a_3, b_3 \neq 0$, in which case both input vectors compute $x_1 + x_2$, which is a superadditive function; \textbf{(3)} without loss of generality, $a_3 = 0$ and $b_3 \neq 0$, in which case $f(\vec{a}) + f(\vec{b}) = \min(a_1, a_2) + b_1 + b_2 \leq a_1 + a_2 + b_1 + b_2 = f(\vec{a} + \vec{b})$.
The function is positive-continuous, since the only points of discontinuity are when $x_3$ changes from zero to positive.
The function is piecewise rational linear, since $\min$ is piecewise rational linear.

\begin{theorem}\label{main-theorem}
A function $f : \mathbb{R}^n_{\geq0} \rightarrow \mathbb{R}_{\geq0}$ is computable by a composable CRC if and only if it is superadditive positive-continuous piecewise rational linear.
\end{theorem}
We prove each direction of the theorem independently in Sections~\ref{sec:thm_onlyif}~and~\ref{sec:thm_if}.

\subsection{Computable Functions are Superadditive Positive-Continuous Piecewise Rational Linear}\label{sec:thm_onlyif}

Here, we prove that a stably computable function must be superadditive positive-continuous piecewise rational linear.
The constraints of positive-continuity and piecewise rational linearity stem from previous work:

\begin{lemma}\label{lem:pos-cont-pw-rl}[Proven in~\cite{rate-indep}]
If a function $f : \mathbb{R}^n_{\geq0} \rightarrow \mathbb{R}_{\geq0}$ is stably computable by a CRC, then $f$ is positive-continuous piecewise rational linear.
\end{lemma}

In addition to the constraints in the above lemma, we show in Lemma~\ref{lem:superadditive} that a function must be superadditive if it is stably computed by a CRC.
To prove this, we first note a useful property of reachability in CRNs. 
\begin{lemma}\label{lem:atob_actobc}
Given states $\vec{a},\vec{b},\vec{c}$, if $\vec{a} \rightarrow \vec{b}$ then $\vec{a} + \vec{c} \rightarrow \vec{b} + \vec{c}$.
\end{lemma}
\begin{proof}Adding species cannot prevent reactions from occurring.
Thus, we can consider the series of reactions where $\vec{c}$ doesn't react to reach the state $\vec{b} + \vec{c}$ from the state $\vec{a} + \vec{c}$.
\end{proof}
\noindent 
We now utilize this lemma to prove that composably computable functions must be superadditive.

\begin{lemma}\label{lem:superadditive}
If a function $f : \mathbb{R}^n_{\geq0} \rightarrow \mathbb{R}_{\geq0}$ is stably computable by a composable CRC, then $f$ is superadditive.
\end{lemma}
\begin{proof}
Assume $\mathcal{C}$ stably computes $f$.
By definition of $\mathcal{C}$ stably computing $f$, $\forall$ initial states $\vec{x}_1, \vec{x_2}$, $\exists\ \vec{o}_1, \vec{o}_2$ such that $\vec{x}_1 \rightarrow \vec{o}_1$ with $\vec{o}_1(Y) = f(\vec{x}_1)$ and $\vec{x}_2 \rightarrow \vec{o}_2$ with $\vec{o}_2(Y) = f(\vec{x}_2)$.
Consider $\mathcal{C}$ on input $\vec{x}_1 + \vec{x}_2$.
By Lemma~\ref{lem:atob_actobc}, $\vec{x}_1 + \vec{x}_2 \rightarrow \vec{o}_1 + \vec{x}_2$, and again by Lemma~\ref{lem:atob_actobc}, $\vec{o}_1 + \vec{x}_2 \rightarrow \vec{o}_1 + \vec{o}_2$.
Looking at the concentration of output species $Y$, we have $(\vec{o}_1 + \vec{o}_2)(Y) = f(\vec{x}_1) + f(\vec{x}_2)$.
Since $\calC$ stably computes $f$, there exists an output stable state $\vec{o}'$ reachable from initial state $\vec{x}_1 + \vec{x}_2$ and reachable from state $\vec{o}_1 + \vec{o}_2$, with $\vec{o'}(Y) = f(\vec{x}_1 + \vec{x}_2)$.
Since $\mathcal{C}$ is composable, we can assume that species $Y$ does not appear as a reactant without loss of generality by Lemma~\ref{lem:comp-implies-not-using-out}. Thus the concentration of $Y$ in any state reachable from state $\vec{o}_1 + \vec{o}_2$ cannot be reduced from $f(\vec{x}_1) + f(\vec{x}_2)$, implying $\vec{o'}(Y) = f(\vec{x}_1 + \vec{x}_2) \geq f(\vec{x}_1) + f(\vec{x}_2)$.
This holds for all input states $\vec{x}_1$, $\vec{x}_2$, and thus $f$ is superadditive.
\end{proof}

\begin{corollary}\label{cor:no-max}
No composable CRC computes $f(x_1, x_2) = \max(x_1,x_2)$.
\end{corollary}

\subsection{Superadditive Positive-Continuous Piecewise Rational Linear Functions are Computable}\label{sec:thm_if}

It was shown in \cite{ovchinnikov2002max} that every piecewise linear function can be written as a $\max$ of $\min$s of linear functions. This fact was exploited in \cite{rate-indep} to construct a CRN that dual-rail 
computed continuous piecewise rational linear functions. To directly compute a positive-continuous piecewise rational linear function, dual-rail networks were used to compute the function on each domain, take the appropriate $\max$ of $\min$s, and then the reaction $Y^+ + Y^- \to \emptyset$ was used to convert the dual-rail output into a direct output where the output species is $Y^+$. However, this technique is not usable in our case: by Corollary~\ref{cor:no-max}, we cannot compute the $\max$ function, and the technique of converting dual-rail output to a direct output is not output-oblivious. In fact, computing $f(Y^+, Y^-) = Y^+ - Y^-$ is not superadditive, and so by Lemma~\ref{lem:superadditive}, there is no composable CRC which computes this conversion.

Since our functions are positive-continuous, we first consider domains where the function is continuous, and show that it can be computed by composing rational linear functions with $\min$.
Since rational linear functions and $\min$ can be computed without using the output species as a reactant, we achieve composability.
We then extend this argument to handle discontinuities between domains.

\begin{definition}
An \emph{open ray} $\ell$ in $\mathbb{R}^n$ from the origin through a point $\vec{x}$ is the set $\ell = \set{\vec{y} \in \mathbb{R}^n \ |\ \vec{y} = t\cdot\vec{x},\ t \in \mathbb{R}_{>0}}$. Note that $t$ is strictly positive, so the origin is not contained in $\ell$.
\end{definition}

\begin{definition}
We call a subset $D \subseteq \mathbb{R}^n$ a \emph{cone} if for all $\vec{x} \in \mathbb{R}^n$, we know that $\vec{x} \in D$ implies the open ray from the origin through $\vec{x}$ is contained in $D$.
\end{definition}

\begin{lemma}\label{lem:domains-are-cones}
Suppose we are given a continuous piecewise rational linear function $f: \mathbb{R}_{>0}^n \to \mathbb{R}_{\ge 0}$. Then we can choose domains for $f$ which are cones which contain an open ball of non-zero radius.
\end{lemma}

Intuitively, we can consider any open ray from the origin and look at the domains for $f$ along this ray. If the ray traveled through different domains, then there must be boundary points where the function switches domains. But we know that $f$ is continuous, so the domains must agree on their boundaries. Since there is only one line that passes through the origin and any given point, the domains must share the same linear function to be continuous. Thus we can place the ray into one domain corresponding to its linear function.
Applying this argument to all rays gives these domains as cones. This argument is formalized in a proof in the appendix.

\begin{lemma}\label{lem:fn-is-a-min}
Any superadditive continuous piecewise rational linear function $f: \mathbb{R}_{>0}^n \to \mathbb{R}_{\ge 0}$ can be written as the minimum of a finite number of rational linear functions $g_i$.
\end{lemma}

\begin{proof}
Since $f$ is a continuous piecewise rational linear function, by Lemma~\ref{lem:domains-are-cones}, we can choose domains $\set{D_i}_{i = 1}^N$ for $f$ which are cones and contain an open ball of non-zero radius, such that $f|_{D_i} = g_i|_{D_i}$, where $g_i$ is a rational linear function. Now pick any $\vec{x} \in \mathbb{R}_{>0}^n$ and any $g_j$. Then because $D_j$ is a cone containing an open ball of finite radius, it contains open balls with arbitrarily large radii. In particular, it contains a ball with radius greater than $|\vec{x}|$, so there exist points $\vec{y}, \vec{z} \in D_j$ such that $\vec{y} + \vec{x} = \vec{z}$. By the superadditivity of $f$, the linearity of $g_j$, and the fact that $\vec{y}, \vec{z} \in D_j$, we see: 
\[g_j(\vec{y}) + f(\vec{x}) = f(\vec{y}) + f(\vec{x}) \le f(\vec{z}) = g_j(\vec{x} + \vec{y}) = g_j(\vec{y}) + g_j(\vec{x})\]

\noindent so that $f(\vec{x}) \le g_j(\vec{x})$. Since this is true for all $g_j$, and since we know that $f(\vec{x}) = g_i(\vec{x})$ for some $i$, we see that $f(\vec{x}) = \min_i g_i(\vec{x})$, as desired.
\end{proof}

Lemma~\ref{lem:fn-is-a-min} is particularly useful for us since, as seen in the introduction, CRCs computing $\min$ are easy to construct, and rational linear functions are relatively straightforward as well.
The next lemma gives details on constructing a CRC to compute $f$ by piecing together CRCs which compute the components (rational linear functions) of $f$ and then computing the $\min$ across their outputs.
However, since Lemma~\ref{lem:fn-is-a-min} as given applies to continuous functions with domain $\R^n_{>0}$, so does this lemma; we handle the domain $\R^n_{\geq0}$ later on.

\begin{lemma}\label{lem:linear-fn-construction}
We can construct a composable CRC that stably computes any superadditive continuous piecewise rational linear function $f : \mathbb{R}^n_{>0} \rightarrow \mathbb{R}_{\geq 0}$.
\end{lemma}
\begin{proof}
By Lemma~\ref{lem:fn-is-a-min}, we know that $f$ can be written as the minimum of a finite number of rational linear functions $g_i$. Observe that a general rational linear function $g(\vec{x}) = a_1 x_1 + a_2 x_2 + \ldots a_n x_n$ is stably computed by the reactions $$\forall i,\; k_iX_i \rightarrow a_ik_i Y$$ where $k_i$ is a positive integer such that $k_ia_i$ is also a positive integer. Since $f$ is the minimum of a number of $g_i$'s, we can make a chemical reaction network where we compute each $g_i$ using a copy of the input species, calling the output $Y_i$ (the reaction $X_1 \rightarrow X_1^1 + \ldots + X_1^5$ produces five species with concentrations equal to $X_1$'s initial concentration, effectively copying the input species so that the input may be a reactant in several modules without those modules competing). Next, we use the chemical reaction $$Y_1 + \ldots + Y_n \rightarrow Y$$ to get the minimum of the $Y_i$'s. Since each $Y_i$ obtains the count of the corresponding $g_i$, this CRN will produce the minimum of the $g_i$'s quantity of Y's. Thus, according to Lemma~\ref{lem:fn-is-a-min}, the described CRC stably computes $f$.
Note that each sub-CRC described in this construction is output-oblivious, and thus composable, so the composition of these modules maintains correctness.
\end{proof}

The above construction only handles the domain $\R^n_{>0}$, where we know our functions are continuous by positive-continuity. However, when extended to the domain $\R^n_{\geq0}$, positive-continuity of our functions allows discontinuity where inputs change from zero to positive.
The challenge, then, is to compute the superadditive \emph{continuous} piecewise rational linear function corresponding to which inputs are nonzero.

Surprisingly, Lemma~\ref{lem:multiple-domains} below shows that we can express a superadditive positive-continuous piecewise rational linear function as a $\min$ of superadditive \emph{continuous piecewise rational linear functions}.
The first step towards this expression is to see that, given two subspaces of inputs wherein the species present in one subspace $A$ are a superset of the species present in a subspace $B$, the function as defined on the subspace $A$ must be greater than the function as defined on the subspace $B$; otherwise, the function would disobey monotonicity and thus superadditivity, as proven below:

\begin{lemma}\label{lem:domain-inequality}
Consider any superadditive positive-continuous piecewise rational linear function $f: \mathbb{R}_{\geq0}^n \to \mathbb{R}_{\ge 0}$. Write $N = [n]$, and for each $S \subseteq N$, let $g_S(\vec{x})$ be the superadditive continuous piecewise rational linear function that is equal to $f$ on $D_S$. If $S, T \subseteq N$ and $S \subseteq T$, then for all $\vec{x} \in D_S$ we know $g_S(\vec{x}) \le g_T(\vec{x})$.
\end{lemma}

\begin{proof}
Write $\vec{e}_i$ for the vector of length 1 pointing in the positive direction of the $i$th coordinate axis. Define the vector $\vec{v} = \sum_{i \in T \setminus S} \vec{e}_i$. Then for any $\vec{x} \in D_S$ and any $\epsilon \in \mathbb{R}_{>0}$, we know that $\vec{x} + \epsilon\vec{v} \in D_T$. Since $f$ is superadditive, it is also monotonic. Suppose that $g_T(\vec{x}) < g_S(\vec{x})$. Because $g_T$ is continuous, taking $\delta = g_S(\vec{x}) - g_T(\vec{x}) > 0$, there is some small enough $\epsilon > 0$ such that \[f(\vec{x} + \epsilon\vec{v}) = g_T(\vec{x} + \epsilon\vec{v}) < g_T(\vec{x}) + \delta = g_S(\vec{x}) = f(\vec{x})\]
contradicting the monotonicity of $f$. Our assumption must be false, so $g_S(\vec{x}) \le g_T(\vec{x})$.
\end{proof}

%
%
%

Next we define a predicate for each subset of inputs which is true if all inputs in that subset are nonzero. 
Intuitively, in the CRC construction to follow, this predicate is used by the CRC to determine which inputs are present:

\begin{definition}\label{def:s-prediate}
For any set $S \subseteq [n]$, define the \emph{$S$-predicate} $P_S: \mathbb{R}^n_{\ge 0} \to \set{0,1}$ to be the function given by:

\[P_S(\vec{x}) = \begin{cases}
1 & \vec{x}(i) > 0\ \forall i \in S \\
0 & \textnormal{otherwise.}
\end{cases}\]
\end{definition}

A na\"ive approach might be the following: for each subdomain $D_S$, the function is continuous, so compute it by CRC according to Lemma~\ref{lem:linear-fn-construction}, producing an output $Y_S$.
Then compute the $P_S$ predicate by CRC, and if the predicate is true (e.g., a species representing $P_S$ has nonzero concentration), use that species to catalyze a reaction which changes the $Y_S$ to $Y$, the final output of the system.
However, note that if $T$ is a subset of $S$, $P_S$ and and $P_T$ are both true, so this technique will overproduce $Y$.

The following technique solves this issue by identifying a min which can be taken over the intermediate outputs $Y_S$.
In particular, for each $S$, we compute $g_S(\vec{x}) + \sum_{K \not \subseteq S} P_K(\vec{x}) g_K(\vec{x})$, and then take the min of these terms.
When $S$ corresponds to the set of input species with initially nonzero concentrations, then the summation term in this expression is $0$, since $P_K(x) = 0$ for all $K\not\subseteq S$.
When $S$ does not correspond to the set of input species with initially nonzero concentration, then either \textbf{(1)} it is a superset of the correct set $I$, in which case Lemma~\ref{lem:domain-inequality} says that $g_S(\vec{x}) \geq g_I(\vec{x})$ (thus the min of these is $g_I(\vec{x})$) or \textbf{(2)} the summation term added to $g_S(\vec{x})$ contains at least $g_I(\vec{x})$, and since $g_S(x) + g_I(\vec{x}) \geq g_I(\vec{x})$, the min of these is $g_I(\vec{x})$.
Thus taking the min for all $S$ of $g_S(\vec{x}) + \sum_{K \not \subseteq S} P_K(\vec{x}) g_K(\vec{x})$ is exactly $g_I(\vec{x})$, where $I$ is the correct set of initially present input species.

\begin{lemma}\label{lem:multiple-domains}
Consider any superadditive positive-continuous piecewise rational linear function $f: \mathbb{R}_{\geq0}^n \to \mathbb{R}_{\ge 0}$. Write $N = [n]$, and for each $S \subseteq N$, let $g_S(\vec{x})$ be the superadditive continuous piecewise rational linear function that is equal to $f$ on $D_S$. Then, $f(\vec{x}) = \min\limits_{S \subseteq N}[g_S(\vec{x}) + \sum\limits_{K \not \subseteq S} P_K(\vec{x}) g_K(\vec{x})]$.
\end{lemma}

\begin{proof}
For $S \subseteq N$, let $h_S: \R_{\geq 0}^n \to \R_{\geq 0}$ be given by 

\[h_S(\vec{x}) = g_S(\vec{x}) + \sum\limits_{K \not \subseteq S} P_K(\vec{x})g_K(\vec{x}).\]

We want to show that $f(\vec{x}) = \min_{S \subseteq N}h_S(\vec{x})$. To do this, fix $\vec{x} \in \mathbb{R}_{\ge 0}^n$ and define the set
$I = \set{i \in N \ |\ \vec{x}(i) > 0}$. First, let's show that $h_I(\vec{x}) = f(\vec{x})$. By the definition of $I$, for all $K \not \subseteq I$, we know $P_K(\vec{x}) = 0$. Thus, $\sum_{K \not \subseteq I} P_K(\vec{x}) g_K(\vec{x}) = 0$, so $h_I(\vec{x}) = g_I(\vec{x}) = f(\vec{x})$. Now we must show that $h_S(\vec{x}) \ge f(\vec{x})$ for all $S \subseteq N$. There are two cases to consider:
\\\\
Case 1: $S \not\supseteq I$
\\
In this case,
\begin{align}
h_S(\vec{x}) &= g_S(\vec{x}) + \sum_{K \not \subseteq S} P_K(\vec{x}) g_K(\vec{x})\\ &\geq g_S(\vec{x}) + P_I(\vec{x}) g_I(\vec{x})\\ &\geq P_I(\vec{x}) g_I(\vec{x}).
\end{align}
By the definition of $I$, we know $P_I(\vec{x}) = 1$, so $P_I(\vec{x})g_I(\vec{x}) = g_I(\vec{x}) = f(\vec{x})$. Thus we get that $h_S(\vec{x}) \geq f(\vec{x})$.
\\\\
Case 2: $S \supseteq I$
\\
By Lemma~\ref{lem:domain-inequality}, $g_S(\vec{x}) \geq g_I(\vec{x})$. As a result,
\begin{align}h_S(\vec{x}) &= g_S(\vec{x}) + \sum\limits_{K \not \subseteq S} P_K(\vec{x}) g_K(\vec{x})\\ &\ge g_S(\vec{x}) \ge g_I(\vec{x})\\ &= f(\vec{x}).
\end{align}

Since for all $\vec{x} \in \mathbb{R}_{\ge 0}^n$, we know $h_S(\vec{x}) \ge f(\vec{x})$ for all $S \subseteq N$ and $h_I(\vec{x}) = f(\vec{x})$ for some $I \subseteq N$, it follows that $f(\vec{x}) = \min_{S \subseteq N} h_S(\vec{x})$.
\end{proof}

Lemma~\ref{lem:full-construction} takes the above Lemma~\ref{lem:multiple-domains} along with the construction which stably computes on strictly continuous domains from Lemma~\ref{lem:linear-fn-construction} to construct a CRC which stably computes on positive-continuous domains.

\begin{lemma}\label{lem:full-construction}
Given any superadditive positive-continuous piecewise rational linear function $f: \mathbb{R}_{\geq0}^n \to \mathbb{R}_{\geq0}$, there exists a composable CRC which stably computes $f$.
\end{lemma}

\begin{proof}
The proof follows by identifying that the function can be expressed as a composition of functions (via Lemma~\ref{lem:multiple-domains}) which are computable by output-oblivious CRCs and are thus composable by Lemma~\ref{lem:not-using-out-imp-comp}.
By Lemma~\ref{lem:multiple-domains}, we know that $f(\vec{x}) = \min\limits_{S \subseteq N}[g_S(\vec{x}) + \sum\limits_{K \not \subseteq S} P_K(\vec{x}) g_K(\vec{x})]$. 
The first subroutine copies the input species, e.g. $X_1 \rightarrow X_1^1 +\ldots+ X_1^5$, in order for each sub-CRC to not compete for input species.
This copying is output-oblivious.
Then for any $Q \subseteq [n]$, $P_Q(x)$ is computed using one set of copies via the reaction: $$\sum\limits_{i \in Q} X_i \rightarrow P_Q$$
noting that although the predicate $P_Q(x)$ is defined to be $0$ or $1$, it is sufficient in this construction for the concentration of the species representing $P_Q(x)$ to be zero or nonzero. This CRC is output-oblivious.

 We can also compute each $g_Q(\vec{x})$ (via Lemma~\ref{lem:linear-fn-construction}) using copies of the input molecules. 
 This construction is output-oblivious. To compute $P_Q(x)g_Q(x)$ given the concentration species $P_Q$ as nonzero iff $P_Q(x) = 1$ as shown above, we simply compute the following (assuming $Y_{Q}$ is the output of the module computing $g_Q(x)$):
 $$
 f(P_Q, Y_{Q}) = \begin{cases} 
      Y_{Q} & P_Q \neq 0 \\
      0 & P_Q = 0, \\
   \end{cases}
 $$
 which is computed by this output-oblivious CRC: $$Y_{Q} + P_Q \rightarrow Y + P_Q.$$
The CRC computing min is output-oblivious, as seen in the introduction.
The CRC computing the sum of its inputs is output-oblivious (e.g., $X_1 \rightarrow Y, X_2 \rightarrow Y$ computes $X_1 + X_2$).
Since each CRC shown is output-oblivious and thus composable, we can compose the modules described to construct a CRC stably computing $\min\limits_{S \subseteq N}[g_S(\vec{x}) + \sum\limits_{K \not \subseteq S} P_K(\vec{x}) g_K(\vec{x})]$, which is equal to $f(\vec{x})$ by Lemma~\ref{lem:multiple-domains}.
\end{proof}

\begin{corollary}
Given any superadditive positive-continuous piecewise rational linear function $f: \mathbb{R}_{\geq0}^n \to \mathbb{R}_{\geq0}$, there exists a composable CRC with reactions with at most two reactants and at most two products which stably computes $f$.
\end{corollary}
To deduce this corollary, note that the reactions with more than two reactants and/or products are used to compute the following functions: computation of a rational linear function, copying inputs, min, and predicate computation.
We can decompose such reactions into a set of bimolecular reactions. 
For example, a reaction $X_1 + \ldots + X_n \rightarrow Y_1 +  \ldots + Y_n$ can be decomposed into the reactions
$X_1 + X_2 \rightarrow X_{12}$, $X_{12} + X_3 \rightarrow X_{123}, \ldots, X_{123\ldots n-1} + X_n \rightarrow Y_{12\ldots n-1} + Y_n$, $Y_{12\ldots n-1} \rightarrow Y_{12\ldots n-2} + Y_{n-1}, \ldots, Y_{12} \rightarrow Y_1 + Y_2$.
We can verify that each affected module stably computes correctly with these expanded systems of reactions, and remains composable.




\section{Example}

In this section, we demonstrate the construction presented in the previous section through an example.
Consider the function shown in Equation~\ref{eq:example} in Section~\ref{sec:functions}.
As shown in that section, the function is superadditive, positive-continuous, and piecewise rational linear.
Thus, we can apply our construction to generate a composable CRN stably computing this function. Note that while this CRN is generated from our methodology, we have removed irrelevant species and reactions.
   
\begin{multicols}{2}
\noindent Making copies of input:
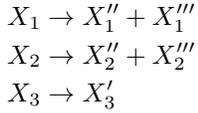
\begin{align}  
   X_1 &\rightarrow X_1'' + X_1'''\\
   X_2 &\rightarrow X_2'' + X_2'''\\
   X_3 &\rightarrow X_3'
\end{align}
Using $X_3'$ to make $P_{3}$, which catalyzes reactions for the domain $X_3 > 0$: 
\begin{align}  
   X_3' &\rightarrow P_{\{3\}}
\end{align}
Computing the sum in $Y_{\{3\}}$:
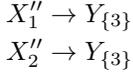
\begin{align}
   X_1'' &\rightarrow Y_{\{3\}}\\
   X_2'' &\rightarrow Y_{\{3\}}
\end{align}
Computing the min in $Y_\emptyset$:
\begin{align}
   X_1''' + X_2''' &\rightarrow Y_\emptyset
\end{align}
Making a copy of $Y_{\{3\}}$ for use in increasing $Y_\emptyset'$:
\begin{align}
   Y_{\{3\}} &\rightarrow Y_{\{3\}}' + Y_{\{3\}, \emptyset}
\end{align}
Increase $Y_\emptyset'$ so that it will not be the min when $x_3$ is present:
\begin{align}
   Y_{\{3\}, \emptyset} + P_{\{3\}} &\rightarrow Y_\emptyset' + P_{\{3\}}
\end{align}
Rename $Y_\emptyset$ to $Y_\emptyset'$ so that it will be summed with the term created by the previous reaction:
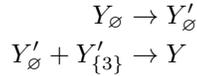
\begin{align}
   Y_\emptyset &\rightarrow Y_\emptyset'\\
   Y_\emptyset' + Y_{\{3\}}' &\rightarrow Y
\end{align}
\end{multicols}

\section{Functions Computable by Composable CRNs with Initial Context}\label{sec:initial_context}
So far, our CRCs restrict the concentrations of non-input species in the initial state to be zero.
One may consider some (constant) initial concentration of non-input species, called \emph{initial context}, and how that may affect computation.

\begin{definition}\label{def:initial-context}
A CRC with initial context is denoted $\calC^{I,\vec{i}} = (\Lambda,R,\Sigma,Y,I, \vec{i})$
with $\Lambda, R, \Sigma,$ and $Y$ defined as in Definition~\ref{def:CRC}, and the initial context species $I \subset \Lambda \setminus ( \Sigma \cup Y )$ and initial context concentrations $\vec{i} \in \R^{I}_{\geq0}$.
$\calC^{I,\vec{i}}$ stably computes $f : \Rp^n \to \Rp$ if, for all $\vx \in \Rp^n$ and all $\vc$ such that $\vx + \vec{i} \to \vc$, there exists an output stable state $\vo$ such that $\vc \to \vo$ and $\vo(Y) = f(\vx)$.
\end{definition}\label{def:initial-context}
We will show that such CRCs can compute exactly what we could already compute when we set one of our inputs to one. Note that this results in a larger class of computable functions than superadditive positive-continuous piecewise rational linear. For example, we can now compute the non superadditive function $f(x_1) = 1$. To formalize this idea, we will also need the following definition.
\begin{definition}\label{def:linear-extension}


Consider a piecewise affine (linear with an offset) function $f: \Rp^n \to \Rp$ defined by $g_i$ on domain $D_i$.
The linear extension $f': \Rp^n \times \mathbb{R}_{>0}\to \Rp$ of $f$ is constructed as follows: for all $\vec{x} \in \Rp^n, \gamma \in \mathbb{R}_{>0}$ with $\vec{x}/\gamma \in D_i$, let $f'(\vec{x},\gamma) = g_i(\vec{x}) - g_i(\vec{0}) + g_i(\vec{0}) \gamma$.


\end{definition}
For example, given the piecewise affine function $$f(x_1) = \begin{cases} 1 & x_1 > 1 \\ x_1 & x_1 \leq 1 \end{cases}$$ its linear extension would be $$f'(x_1,\gamma) = \begin{cases} \gamma & x_1 > \gamma \\ x_1 & x_1 \leq \gamma \end{cases}.$$ Observe that $f'$ is just the minimum of $x_1$ and $\gamma$.
Intuitively the linear extension replaces all constant terms with a linear term in the new variable $\gamma$.
We show that initial context for composable CRCs allows only functions whose linear extensions are superadditive, positive-continuous, piecewise rational linear. Intuitively, this gives us the same class of functions that we would have without initial context when we set one of our inputs to the constant value one.

As defined, we allow an arbitrary number of initial context species with differing initial concentrations, but we will focus on the single species case with an initial concentration of one.
This is well motivated: we show one initial context species with concentration one is equivalent in stable computing power to having any number of species with nonnegative rational initial concentrations.

\begin{lemma}\label{lem:equiv-initial-context}
Given a CRC with initial context $\calC^{I,\vec{i}}$ with $\vec{i} \in \mathbb{Q}^{I}_{\geq0}$ (rational initial concentrations) which stably computes $f$, there exists a CRC $\calC^{I',\vec{i'}}$
with $I' = \{S'\}$ and $\vec{i'}(S') = 1$ (one initial species with concentration one) which stably computes $f$. 
\end{lemma}

\begin{proof}
Let $q_i = \frac{a_i}{b_i}$ for $a_i,b_i \in \mathbb{Z}_{\geq 0}$ be the initial (rational) concentrations of the initial context in $\calC^{I,\vec{i}}$.
Observe that the CRN:
$$S' \rightarrow S_1' + S_2' + \ldots + S_k'$$
can be used to produce $k$ species with concentrations equal to $S'$'s initial concentration.
Then the CRN:
$$bS_i' \rightarrow aS_i$$
for each $i$ produces a concentration of $q_i$ for species $S_i$. 
\end{proof}
While this schema cannot be used to generate initial context with irrational concentrations, continued fractions can be used to approximate irrational numbers as rational numbers with arbitrarily small error. Thus our restriction to one initial species with a initial concentration one is reasonable to consider for stable computation in this model.
To characterize the functions stably computable with initial context, we first prove some lemmas.
Recall $\text{Post}(\vec{c})$ is the set of states reachable from $\vec{c}$, i.e., $\{\vec{d} \ | \ \vec{c}\rightarrow \vec{d} \}$.

\begin{lemma}\label{lem:const-scaling-reachable-states}
Given a CRN $\{\Lambda, R\}$, for any $\gamma \in \mathbb{R}_{\geq 0}$ and $\vec{c} \in \Rp^{\Lambda}$, $\mathrm{Post}(\gamma\vec{c}) = \gamma\mathrm{Post}(\vec{c})$.
\end{lemma}
\begin{proof}
If $\vec{c} \to \vec{d}$, then for any $k\in \mathbb{R}_{\geq 0}$, $k\vec{c} \to k\vec{d}$. This can be verified by taking the straight line segments to get from $\vec{c}$ to $\vec{d}$ and scaling them by $k$.

When $\gamma=0$, this lemma is trivial, as the only state reachable from the zero state is the zero state and zero times any state yields the zero state. Thus we only need to consider the case where $\gamma > 0$.

Let $\vec{v} \in \text{Post}(\gamma\vec{c})$. By the definition of $\text{Post}$ this implies that $\gamma\vec{c} \to \vec{v}$. This implies that $\vec{c} \to \frac{1}{\gamma} \vec{v}$. Therefore $\frac{1}{\gamma} \vec{v} \in \text{Post}(\vec{c})$ and $\vec{v} \in \gamma\text{Post}(\vec{c})$. 

Let $\vec{v} \in \gamma\text{Post}(\vec{c})$. By the definition of $\text{Post}$ this implies that $\vec{c} \to \frac{1}{\gamma} \vec{v}$. This implies that $\gamma\vec{c} \to \vec{v}$. Which implies that $\vec{v} \in \text{Post}(\gamma\vec{c})$.
\end{proof}
\begin{lemma}\label{lem:const-scale-ini-ctx-behavior}
Let $\calC^{I,\vec{i}}$ with $I = \{S_1\}$ and $\vec{i}(S_1) = q_1$ stably compute $f$. Then $\calC^{I,\gamma\vec{i}}$ for $\gamma \in \R_{>0}$ stably computes some function $g$.
\end{lemma}
\begin{proof}
Consider running $\calC$ with an initial concentration of $q_1' = \gamma q_1$ for the species $S_1$. Observe that the initial state $\vec{x'} = \frac{1}{\gamma} \vec{x}$, where $\vec{x}$ is a state that will stably compute $f$ under the definition of $\calC$. By Lemma~\ref{lem:const-scaling-reachable-states}, we know that the set of reachable states from $\vec{x'}$ is equal to the set of reachable states from $\vec{x}$ scaled by $\gamma$. Thus consider some state $\vec{c'}$ reached from $\vec{x'}$. Observe that there exists a state $\vec{c}$ reachable from $\vec{x}$ such that $\vec{c'} = \gamma \vec{c}$. Consider some output stable state $\vec{o}$ reachable from $\vec{c}$. Observe that by Lemma~\ref{lem:const-scaling-reachable-states} the state $\vec{o'} = \gamma \vec{o}$ must be reachable from $\vec{c'}$. Likewise by Lemma~\ref{lem:const-scaling-reachable-states} we know that $\vec{o'}$ must be an output stable state. Thus, we know that $\calC$ must stably compute some function regardless of the initial value for $S_1$.
\end{proof}

We thus know that scaling the value of the initial context retains the fact that $\calC$ stably computes a function in the region where that species has a positive concentration.
We can then view a single species of initial context as an additional input to $\calC$ and claim that this CRN stably computes some function, using lemmas from the case of no initial context to prove properties of that function.
\begin{lemma}\label{lem:linear-extension}
Let $\mathcal{C}^{I,\vec{i}}$ be output-oblivious with $I = \{S_1\}$ and $\vec{i}(S_1) = 1$. $\mathcal{C}^{I,\vec{i}}$ stably computes $f : \mathbb{R}_{\geq 0}^n \to \mathbb{R}_{\geq 0}$ only if $f$ is a piecewise affine function whose linear extension is superadditive positive-continuous piecewise rational linear.
\end{lemma}
\begin{proof}
Intuitively, we treat the initial context as an input species. 
For $\mathcal{C}^{I,\vec{i}} = (\Lambda, R,\Sigma, Y, I, \vec{i})$, let $\mathcal{C}' = (\Lambda, R, \Sigma \cup I, Y)$.
By assumption, $\mathcal{C}^{I,\vec{i}}$ is output-oblivious, so $\mathcal{C}'$ is output-oblivious.
Further, by Lemma~\ref{lem:const-scale-ini-ctx-behavior}, $\mathcal{C}'$ stably computes some function for all (positive) initial concentrations of $S_1$.
So $\mathcal{C}'$ must stably compute a superadditive, positive-continuous, piecewise rational linear function $f'$ on all domains with positive (nonzero) initial concentration of $S_1$. Since $\mathcal{C}^{I,\vec{i}}$ and $\mathcal{C}'$ share the same CRN, output species, and input species other than $S_i$, we know that for any $\vec{x} \in \mathbb{R}_{\geq 0}^{\Sigma}$, $f'(\vec{x},1) = f(\vec{x})$.
All that remains is to show that $f'$ is the linear extension of $f$. Pick any $\gamma \in \mathbb{R}_{>0}$, $\vec{x} \in \Rp^n$  such that $\vec{x} / \gamma \in D_i$. Let $g_i$ be the affine function used by $f$ on $D_i$. 
Since $f'(\vec{x}/\gamma,1) = f(\vec{x}/\gamma)$ and the domains of $f'$ are cones by Lemma~\ref{lem:domains-are-cones}, $f'(\vec{x},\gamma) = \gamma g_i(\vec{x} / \gamma)$. Expanding out the right hand side gives us that $ \gamma f'(\vec{x}/\gamma,1) = g_i(\vec{x}) - g_i(\vec{0}) + \gamma g_i(\vec{0})$. By Definition \ref{def:linear-extension} $f'$ is the linear extension of $f$.
\end{proof}

\begin{theorem}
Let $\calC^{I,\vec{i}}$ be output-oblivious with $\vec{i} \in Q^I_{\geq0}$. Then $\calC^{I,\vec{i}}$ stably computes $f : \Rp^n \to \Rp$ if and only if $f$ is a rational affine function whose linear extension is super-additive, positive-continuous, piecewise rational linear.
\end{theorem}
 \begin{proof}
Lemma~\ref{lem:linear-extension} gives us the necessary conditions for composable CRCs with initial context. Now all we need to do is show that there is a composable CRC for any rational affine function $f$ whose linear extension is superadditive, positive-continuous, piecewise rational linear. Since the linear extension of $f$ is superadditive, positive-continuous, piecewise rational linear, Lemma~\ref{lem:full-construction} gives a CRC that computes this function. Observe that changing the input in this CRC representing the linear extension ($\gamma$) to the initial context species will give us a new CRC that computes $f$.
\end{proof}

\section{Discussion}
Instead of continuous concentrations of species, one may consider discrete counts.
This changes which functions are stably computed by CRNs.
Without the composability constraint,~\cite{chen2014deterministic} shows in the discrete model that a function $f : \N^n \rightarrow \N$ is stably computable by a direct CRN if and only if it is semilinear; i.e., its graph $\{(x,y) \in \N^n \times \N \ |\ f(x) = y \}$ is a semilinear subset of $\N^n \times \N$.
The proof that composably computable functions must be superadditive (Lemma~\ref{lem:superadditive}) holds for the discrete model as well.
Additionally, there exist functions which are superadditive and semilinear but are not computable in the discrete model by a composable CRN.
For example (the proof is omitted):
\[ f(x_1, x_2) = \begin{cases} 
       x_1 - 1 & x_1 > x_2 \\
       x_1 & x_1 \leq x_2,
   \end{cases}
\]
so the class of composably computable functions is slightly more restricted.

Three recent works characterized output-oblivious computation in the discrete model.
With initial context, the characterization for functions on two inputs ($f : \N^2 \rightarrow \N$) was given by~\cite{DBLP:conf/opodis/ChuggHC18},
and subsequently extended to arbitrarily many inputs~\cite{severson2019composable}.
More recently, \cite{hashemi2020composable} showed that without initial context, computable functions must be superadditive in addition to the constraints presented in \cite{severson2019composable}.


Our negative and positive results are proven with respect to stable computation, which formalizes our intuitive notion of rate-independent computation.
It is possible to strengthen our positive results to further show that our CRNs converge (as time $t \rightarrow \infty$) to the correct output from any reachable state under \emph{mass-action kinetics} (proof omitted). 
It is interesting to characterize the exact class of rate laws that guarantee similar convergence.

Apart from the dual-rail convention discussed in the introduction, other input/output conventions for computation by CRNs have been studied.
For example, \cite{fractionalEncoding} considers \emph{fractional encoding} in the context of rate-dependent computation.
As shown by dual-rail, different input and output conventions can affect the class of functions stably computable by CRNs. While using any superadditive positive continuous piecewise rational linear output convention gives us no extra computational power---since the construction in this paper shows how to compute such an output convention directly---it is unclear how these conventions change the power of rate-independent CRNs in general.

Finally we can ask what insights the study of composition of rate-independent modules gives for the more general case of rate-dependent computation.
Is there a similar tradeoff between ease of composition and expressiveness for other classes of CRNs?

\bibliography{references.bib}
\bibliographystyle{plain}

\clearpage

\newpage
\section{Appendix}

Most definitions and lemmas in this section work towards proving Lemma~\ref{lem:max}.
We also include a proof of Lemma~\ref{lem:domains-are-cones}.
\begin{definition}
A \emph{polyhedron} is a subset of $\R^n$ of the form $\{\vec{x}\ |\ A\vec{x} \le \vec{b}\}$ for some $m \times n$ matrix $A$ and some vector $\vec{b} \in \R^m$.
\end{definition}

\begin{definition}
A \emph{convex polytope} is the convex hull of a finite set of points in $\R^n$.
\end{definition}

\begin{definition}
A \emph{polyhedral cone} is a set of the form $\{\vec{x} = \lambda_1\vec{x}_1 + \ldots + \lambda_n\vec{x}_n\ |\ \lambda_1,\ldots,\lambda_n \ge 0\}$ for some finite set of points $\set{x_1, \ldots, x_n}$ in $\R^n$
\end{definition}
The following lemma comes from a previous work. Note that this sum is the Minkowski sum.
\begin{lemma}\label{lem:polyhedra} [Proven in \cite{schrijverlinprog}]
A subset $P \subseteq \R^n$ is a polyhedron if and only if $P = Q + C$ for some convex polytope $Q$ and some polyhedral cone $C$.
\end{lemma}

\begin{lemma}\label{lem:polyhedron}
For a given state $\vec{c}$ of a CRN $\mathcal{C}$, the set of states $\{\vec{d}\ |\ \vec{c} \to^1 \vec{d}\}$ that are straight-line reachable from $\vec{c}$ is a polyhedron.
\end{lemma}

\begin{proof}
If $m = |R|$ is the number of reactions in $\mathcal{C}$ and $n = |\Lambda|$ is the number of species in $\mathcal{C}$, then the stoichiometry matrix can be thought of as a linear transformation from the reaction space $\R^m$ to the species space $\R^n$. Let $\mathcal{R}_{\vec{c}}$ be the set of basis vectors corresponding to reactions applicable at $\vec{c}$. Then since $M$ is a linear transformation, it sends the polyhedral cone defined by $\mathcal{R}_{\vec{c}}$ to a polyhedral cone $C'$ in $\R^n$. By Lemma~\ref{lem:polyhedra}, the set $\vec{c} + C'$ is a polyhedron in $\R^n$, and since the set of states that are straight-line reachable from $\vec{c}$ is the intersection of this polyhedron with the set of all vectors with nonnegative components, it is also a polyhedron.
\end{proof}

\begin{definition}
The set of \emph{possible species} produced from a state $\vec{c}$ is 
\[\mathcal{P}(\vec{c}) = \bigcup_{d \in \text{Post}(\vec{c})} [\vec{d}]\]
\end{definition}

\begin{lemma}\label{lem:convex-combination}
For a CRN the set of reachable states is closed under convex combination.
\end{lemma}
\begin{proof}
Consider some state $\vec{c}$ and states $S=\set{\vec{\alpha_1}, \ldots, \vec{\alpha_k}}$ reachable from $\vec{c}$. Let $\vec{d} = \sum_{i=0}^k a_i \alpha_i$ where $\forall i\; a_i>0$ and $\sum_{i=0}^k a_i = 1$. By lemma~\ref{lem:const-scaling-reachable-states} we know that $\forall i\; a_i\vec{\alpha_i}$ is reachable from $a_i \vec{c}$. Since $\sum_{i=0}^k a_i \vec{c} = \vec{c}$, we know that $\vec{c} \to \sum_{i=0}^k a_i \vec{\alpha_i}$, which is equal to $\vec{d}$.
\end{proof}

\begin{lemma}\label{lem:all-pos-exist}
Given a state $\vec{c}$, there is a state $\vec{d}$ reachable from $\vec{c}$ such that $\mathcal{P}(\vec{c}) = [\vec{d}]$. For such a state, $\mathcal{P}(\vec{d}) = \mathcal{P}(\vec{c})$.
\end{lemma}

\begin{proof}
If $\vec{c}$ is the zero vector, observe that $\mathcal{P}(\vec{c}) = [\vec{c}]$, so setting $\vec{d}=\vec{c}$ makes this hold. Otherwise, for each species $S \in \mathcal{P}(\vec{c})$, there is some state $\vec{d}_S$ reachable from $\vec{c}$ with $S \in [\vec{d}_S]$. Then the state $\vec{d} = \frac{1}{|\mathcal{P}(\vec{c})|}\sum_{S \in \mathcal{P}(\vec{c})} \vec{d}_S$ is reachable from $\vec{c}$ by lemma~\ref{lem:convex-combination}. Since $\vec{d}$ contains a positive contribution from each $\vec{d}_S$, $\mathcal{P}(\vec{c}) = [\vec{d}]$. Since $\vec{c} \to \vec{d}$ we know that $\mathcal{P}(\vec{d}) \subseteq \mathcal{P}(\vec{c})$. Since $\mathcal{P}(\vec{c}) = [\vec{d}]$ and $[\vec{d}] \subseteq \mathcal{P}(\vec{d})$ we know that $\mathcal{P}(\vec{c}) \subseteq \mathcal{P}(\vec{d})$. Thus we can conclude that $\mathcal{P}(\vec{d}) = \mathcal{P}(\vec{c})$.

\end{proof}

\begin{lemma}\label{lem:all-pos-prop}
If $\vec{c}$ is a state such that $\mathcal{P}(\vec{c}) = [\vec{c}]$, then any state $\vec{d}$ that is reachable from $\vec{c}$ is straight-line reachable from $\vec{c}$.
\end{lemma}

\begin{proof}
Since $\mathcal{P}(\vec{c}) = [\vec{c}]$, by the definition of $\mathcal{P}(\vec{c})$ we know that the set of applicable reactions from $\vec{c}$ is a superset of those applicable from any state reachable from $\vec{c}$. Thus we can take the sum of all the straight-line segments used to reach $\vec{d}$ from $\vec{c}$ and apply them all in a single straight-line segment to get $\vec{c} \to^1 \vec{d}$.
\end{proof}

\noindent\textbf{Lemma~\ref{lem:max}.} 
\emph{For any state $\vec{c}$ and any species $S$, if the amount of $S$ present in any state reachable from $\vec{c}$ is bounded above, there is a state $\vec{d}$ reachable from $\vec{c}$ such that for any state $\vec{a}$ reachable from $\vec{d}$, we know that $\vec{a}(S) \le \vec{d}(S)$.}

\begin{proof}
By Lemma~\ref{lem:all-pos-exist}, there is some $\vec{c}_1$ reachable from $\vec{c}$ such that $[\vec{c}_1] = \mathcal{P}(\vec{c})$. By Lemma~\ref{lem:polyhedron}, we know that the states that are straight-line reachable from $\vec{c}_1$ are a polyhedron $P$. The linear map $\R^n \to \R$ sending $\vec{x} \mapsto \vec{x}(S)$ maps $P$ to some polyhedral subset of $\R$---in particular this is a closed subset. Since we assume that the image of this map is bounded above, we know that this subset attains its maximum $M$, so there is some $\vec{d} \in P$ with $\vec{d}(S) = M$. Any state $\vec{a}$ that is reachable from $\vec{d}$ is also reachable from $\vec{c}_1$, so by Lemma~\ref{lem:all-pos-prop} it is contained in $P$. As a result, $\vec{a}(S) \le M = \vec{d}(S)$.
\end{proof}

\noindent\textbf{Lemma~\ref{lem:domains-are-cones}.} 
\emph{ Suppose we are given a continuous piecewise rational linear function $f: \mathbb{R}_{>0}^n \to \mathbb{R}_{\ge 0}$. Then we can choose domains for $f$ which are cones which contain an open ball of non-zero radius.}

\begin{proof}
Since $f$ is piecewise rational linear, we can pick a finite set of domains $\mathcal{D} = \set{D_i}_{i = 1}^N$ for $f$, such that $f|_{D_i} = g_i|_{D_i}$, where $g_i$ is a rational linear function. Fix a domain $D_k$, and consider any point $\vec{x} \in D_k$. Since the open ray $\ell_{\vec{x}}$ from the origin passing through $\vec{x}$ is contained in $\mathbb{R}_{>0}^n$, it is covered by the domains in $\mathcal{D}$. If we write any point $\vec{y} \in \ell_{\vec{x}}$ in the form $t \cdot \vec{x}$, then, for each $i$, the restriction of $g_i$ to $D_i \cap \ell_{\vec{x}}$ is of the form $g_i(t \cdot \vec{x}) = \alpha_i t$ for some $\alpha_i \in \mathbb{R}$. Since $\vec{x} \in D_k \cap \ell_{\vec{x}}$, we know that $f(1 \cdot \vec{x}) = \alpha_k \cdot 1 = \alpha_k$

Now suppose that for some $s \in \mathbb{R}_{> 0}$ we know that $f(s \cdot \vec{x}) \neq \alpha_k s$. First consider the case where $s > 1$. Then define the set $A = \set{t \in [1,s] \ |\ f(t\cdot\vec{x}) = \alpha_k t}$ and define the set $B = \set{t \in [1,s] \ |\ f(t\cdot\vec{x}) \neq \alpha_k t}$. We know that $A$ is non-empty since $1 \in A$, so $\sup A$ exists - call it $t'$. From the standard properties of the supremum, we know that there exists a sequence of points $\set{t_j}_{j = 1}^\infty$ such that $t_j \in A$ for all $j$ and $\lim_{j \to \infty} t_j = t'$. As a result, from the continuity of $f$, we see that:

\[f(t' \cdot \vec{x}) = \lim_{j \to \infty} f(t_j \cdot \vec{x}) = \lim_{j \to \infty} \alpha_k t_j = \alpha_k t'\]

So $t' \in A$. However, by assumption, $s \in B$, so that $t' < s$. Since $t'$ is an upper bound on $A$, it must then be the case that $(t', s] \subseteq B$, so that there exists a sequence of points $\set{s_j}_{j = 1}^\infty$ such that $s_j \in B$ for all $j$ and $\lim_{j \to \infty} s_j = t'$. Since there are only finitely many domains in $\mathcal{D}$, but infinitely many $s_j$, by the pigeonhole principle infinitely many of the $s_j$ must be contained in a single domain $D_{k'}$. Now write the subsequence of points contained in $D_{k'}$ as $\set{s_{j'}}_{j' = 1}^\infty$. We still know that $\lim_{j' \to \infty} s_{j'} = t'$, so by the continuity of $f$ and the fact that $s_{j'} \in D_{k'}$, we see that:

\[\alpha_{k}t' = f(t'\cdot\vec{x}) = \lim_{j' \to \infty} f(s_{j'} \cdot\vec{x}) = \lim_{j' \to \infty} \alpha_{k'}s_{j'} = \alpha_{k'}t'\]

Since $t' > 0$, this implies that $\alpha_{k'} = \alpha_k$, so that $f(s_{j'}\cdot \vec{x}) = \alpha_k s_{j'}$. However, this contradicts the fact that we were able to choose $s_{j'} \in B$. As a result, our assumption, that there is some $s > 1$ such that $f(s\cdot \vec{x}) \neq \alpha_k s$, must be false. A similar argument, using the infimum instead of the supremum, shows that there can be no $s < 1$ such that $f(s\cdot \vec{x}) \neq \alpha_k s$. As a result, for every point $t \in \ell_{\vec{x}}$, we know $f(t \cdot\vec{x}) = \alpha_k t$. In other words, $f|_{\ell_{\vec{x}}} = g_k|_{\ell_{\vec{x}}}$, so we can replace $D_k$ with $D_k \cup \ell_{\vec{x}}$ without issue. Doing this for every $\vec{x} \in D_k$, we can replace $D_k$ with a cone. By enlarging every domain in $\mathcal{D}$ in this way, we can choose domains for $f$ which are cones.

Since $f$ is continuous, we can replace each $D_i \in \mathcal{D}$ by its closure, which is again a cone. Suppose that for any $D_i \in \mathcal{D}$, there is a point $\vec{x} \in D_i$ not in the interior of $D_i$. Then $\vec{x}$ is in the closure of the complement of $D_i$, so there exists a sequence $\set{\vec{x}_k}_{k = 1}^\infty$ of points in the complement of $D_i$ such that $\lim_{k \to \infty} \vec{x}_k = \vec{x}$. Since the complement of $D_i$ is covered by the $D_j \in \mathcal{D}$, where $j \not= i$, we know that each $\vec{x}_k$ lies in one of the $D_j$. Since there are only finitely many $D_j$ but infinitely many $\vec{x}_k$, we know that infinitely many $\vec{x}_k$ must lie in at least one of the $D_j$. As a result, $\vec{x}$ is in the closure of this $D_j$, and since $D_j$ is closed, we see that $\vec{x} \in D_j$. Because of this, if $D_i$ has no interior points, then it is completely contained in the other $D_j$, so we can remove it from the set of domains. After doing this for every $D_i$ which contains no interior points, we can ensure that the domains we have chosen for $f$ all contain an open ball of non-zero radius.
\end{proof}

\end{document}